\documentclass[11pt]{article}
\usepackage{geometry, multirow, amsmath, amsfonts, amsthm, amssymb, graphicx, enumerate, natbib, bibentry, float, caption, subcaption, longtable,xcolor,bm}
\usepackage[breaklinks=true]{hyperref} 
\usepackage[nameinlink]{cleveref} 
\usepackage{xcolor} 
\hypersetup{
	bookmarksnumbered=true,
	colorlinks=true,
	anchorcolor={red!50!black},
	linkcolor={red!50!black},
	citecolor={blue!50!black},
	urlcolor={blue!50!black}
}
\usepackage{txfonts}
\usepackage[title]{appendix}
\usepackage{setspace}

\newenvironment{keywords}{
       \list{}{\advance\topsep by0.35cm\relax\small
       \leftmargin=1cm
       \labelwidth=0.35cm
       \listparindent=0.35cm
       \itemindent\listparindent
       \rightmargin\leftmargin}\item[\hskip\labelsep
                                     \bfseries Keywords:]}
     {\endlist}

\newtheorem{Prop}{Proposition}
\newtheorem{Lem}{Lemma}

\newtheorem{Def}{Definition}

\newenvironment{Propprime}[1]
  {\renewcommand{\theProp}{\ref*{#1}$'$}%
   \addtocounter{Prop}{-1}%
   \begin{Prop}}
  {\end{Prop}}

\newenvironment{Defprime}[1]
  {\renewcommand{\theDef}{\ref*{#1}$'$}%
   \addtocounter{Def}{-1}%
   \begin{Def}}
  {\end{Def}}

\let\epsilon\varepsilon

\newcommand\abs[1]{\left|#1 \right|}

\newcommand\Brace[1]{\left\{\begin{aligned}#1\end{aligned}\right.}

\newcommand\kh[1]{\left(#1\right)}
\newcommand\fkh[1]{\left[#1\right]}
\newcommand\hkh[1]{\left\{#1\right\}}

\newcommand\given{\middle|}
\newcommand\eqn[1]{\begin{align}#1\end{align}}
\newcommand\eqns[1]{\begin{align*}#1\end{align*}}
\newcommand\indic[1]{\bm{1}\left(#1\right)}

\begin{document}
\setstretch{1.5}
\begin{titlepage}
\title{\textbf{Hail Mary Pass: Contests with Stochastic Progress}}
\author{Chang Liu\thanks{Department of Economics, Harvard University; \href{mailto:chang_liu@fas.harvard.edu}{\nolinkurl{chang_liu@fas.harvard.edu}}. I am deeply grateful to Shengwu Li, Tomasz Strzalecki, Eric Maskin and Benjamin Golub for their guidance and support throughout the project. I would also like to thank 
Shani Cohen,
Jeffrey Ely,
Yannai Gonczarowski,
Olivier Gossner,
Jack Hirsch,
Zo\"{e} Hitzig,
Johannes H{\"o}rner,
Yiping Hu,
Michihiro Kandori,
Annie Liang,
Ziqi Lu,
Benjamin Niswonger,
Matthew Rabin,
Phillip Strack,
Weijie Zhong,
and especially Alessandro Bonatti for valuable comments and discussion.}}
\date{\begin{tabular}{ rl } 
This version:&May 3, 2023
\end{tabular}}
\maketitle
\thispagestyle{empty}
\begin{abstract}
This paper studies the equilibrium behavior in contests with stochastic progress. Participants have access to a safe action that makes progress deterministically, but they can also take risky moves that stochastically influence their progress towards the goal and thus their relative position. In the unique \textit{well-behaved} Markov perfect equilibrium of this dynamic contest, the follower drops out if the leader establishes a substantial lead, but resorts to ``Hail Marys'' beforehand: no matter how low the return of the risky move is, the follower undertakes in an attempt to catch up. Moreover, if the risky move has a medium return (between high and low), the leader will also adopt it when the follower is close to dropping out \textendash{} an interesting preemptive motive. We also examine the impact of such equilibrium behavior on the optimal prize allocation.
\end{abstract}
\end{titlepage}

\setstretch{1.5}

\newpage

\section{Introduction}
This paper studies contests with stochastic progress induced by risky moves.  Agents exert effort and compete to get ahead. They have the option to make progress deterministically, but they can also undertake risky moves that stochastically affect their relative position. The main theoretical question of this paper is: How do such risky moves affect equilibrium behavior? In particular, under what circumstances should agents take such moves? Will a risky move be adopted no matter how low its return is?

In the baseline model of this paper (Section \ref{sec:model}), two agents participate in a contest to compete for prize. Each agent observes their progress and decides \textit{whether and how} to exert effort. The agents are identical except for their initial position.  If an agent chooses not to exert effort, he irreversibly drops out of the contest. Otherwise, he incurs some flow cost, and selects either the \textit{safe action} that makes progress deterministically at a constant rate, or the \textit{risky move} that stochastically affects his progress. The current leader succeeds with some Poisson rate, at which point the contest ends. The contest is ``winner-takes-all'': all prize is awarded to the winner at the end of the contest.\footnote{In Section \ref{sec:design}, we extend the baseline model by allowing for the possibility that the follower may also have a positive hazard rate of success, and that the loser may also receive a positive prize.} The agents play according to a (pure strategy) \textit{Markov perfect equilibrium (MPE)}, where a Markov strategy  consists of a \textit{stopping rule} specifying \textit{whether} to exert effort, and an \textit{action rule} specifying \textit{how} to exert effort. 

Our main result is that, the follower drops out if the leader establishes a substantial lead, but resorts to ``Hail Marys'' beforehand: no matter how low the return of the risky move is, the follower undertakes it in an attempt to catch up. Moreover, if the risky move has a \textit{medium} return (between \textit{high} and \textit{low}), the leader will also adopt it when the follower is close to dropping out \textendash{} an interesting preemptive motive. 

Toward this result, we first demonstrate three types of MPEs of the dynamic contest and characterize the structure of the equilibrium strategies (Section \ref{sec:MPE}). In each type of MPE, the follower always takes the risky move irrespective of its mean return, but the leader's behavior is critically determined by the return of the risky move compared to the safe action. The analysis boils down to three cases, depending on whether the risky move has a \textit{low}, \textit{medium}, or \textit{high return}. If the risky move has a low return, then the leader will never find it optimal to take it, and the MPE features that each agent takes the risky move if and only if he is the follower (Proposition \ref{prop:L}). In the opposite case where the risky move has a high return, both agents optimally bear risk because it is a worthwhile investment (Proposition \ref{prop:H}). Surprisingly, if the risky move has a {medium return}, i.e., in between the previous two cases, the leader will also adopt it when seeking to force the follower out (Proposition \ref{prop:M}).  Furthermore, we demonstrate that these MPEs are unique within a reasonably large family of MPEs, which we refer to as \textit{well-behaved} MPEs (Propositions \ref{prop:Lprime}, \ref{prop:Hprime} and \ref{prop:Mprime}).

The natural next question is: How does such equilibrium behavior affect the optimal prize allocation? We investigate this question by analyzing an extended version of the baseline model that adapts all the derivations and  intuitions, and allowing the designer to allocate prizes subject to some fixed budget constraint (Secton \ref{sec:design}). We show that three reasonable objectives for the designer are equivalent in our contest model (Proposition \ref{prop:equivobj}). Moreover, the solution to the designer's problem indicates that ``winner-takes-all'' is optimal only when the designer's budget is small  (Proposition \ref{prop:optprize}).

\paragraph{Related Literature} This paper is related to the abundant literature on innovation or patent races pioneered by the work of \cite{Loury79} and \cite{DS80}. Within this literature, two main outcomes are emphasized: (i) \textit{$\epsilon$-preemption} \citep{FGST83,HV85a,HV85b,LM88}, where a slight advantage causes the opponent to drop out immediately, and (ii) \textit{increasing dominance} \citep{GS87,HV87,LM87}, where the follower tends to fall further behind and the leader builds up its advantage. Our model introduces risky moves and examines their impact on equilibrium behavior. The result shows that the follower chooses to stay in the contest when not far behind, hoping to rely on the ``Hail Mary pass'' to catch up. Moreover, the equilibrium degenerates to $\epsilon$-preemption if the variance of the risky move vanishes, demonstrating its essentiality.

Our model is closest to \cite{AB07}, which also studies dynamic competition where the choices of the agents affect the variance of some state variable. In \cite{AB07}, the players can influence the variance, but not the mean, and there is no endogenous exit from the game. In our model, both the mean and the variance may be affected by the agents' actions, and the agents are allowed to drop out at any time. Therefore, the two papers obtain different forms of MPEs. Especially, our model shows that the mean return of the risky move compared to the safe action is crucial in determining the agents' equilibrium behavior.

Technically, our model is a (one-dimensional) stochastic differential game, so it is naturally relevant to the theoretical research in this area  \citep{Girsanov61,Gusein-Zade69,Friedman72,Pliska73,Harris93}. In particular, equilibrium existence may become a tricky issue when the variance of the state variable may be affected by the actions of the players. Fortunately, we directly construct closed-form MPEs in our model, thereby circumventing the existence problem. However, a similar problem arises when we try to establish the uniqueness results, and for now, we have only shown that the constructed MPEs are unique within a reasonably large family of MPEs, which we refer to as \textit{well-behaved} MPEs.

This paper is also related to a recent wave of research on contest design that focuses on feedback in contests \citep{Yildirim05,Ederer10,HKL17,BEM19,BL20,EGKR21}.  Our model abstracts away from information disclosure entirely to focus on the impact of risky moves, but it would be interesting to combine both factors in future research.
\\ \\ 
\indent The rest of the paper is organized as follows. Section \ref{sec:model} lays out the baseline model. The first main part, Section \ref{sec:MPE}, shows that there are three types of MPEs of the dynamic contest and characterize the structure of the equilibrium strategies. The second main part, Section \ref{sec:design}, then analyzes the impact of such behavior on the optimal prize allocation.
 Section \ref{sec:concln} concludes.  Appendix \ref{app:proof} contains the proofs of all results in the main text. 


\section{Model} \label{sec:model}
The baseline model is a continuous-time contest with stochastic progress. Two agents (he)  $i \in\{1,2\}$ participate in a contest to compete for a prize of value $P>0$, and discount at a common rate $r$.

At each instant $t\in[0,\infty)$, each agent observes their \textit{progress} $\kh{k_1\kh{t},k_2\kh{t}}$ and decides \textit{whether and how} to exert effort. The agents are identical except for their initial position $k_i\kh{0}$. 
 If an agent chooses not to exert effort, he irreversibly drops out of the contest. Otherwise, if agent $i$ exerts effort, he incurs flow cost at rate $c$, and selects an action $a_i\kh{t}\in\hkh{0,1}$. The \textit{safe action} $a_i=0$ makes progress deterministically at a constant rate normalized to $1$. The \textit{risky move} $a_i=1$ stochastically affects agent $i$'s progress, with mean return $1+\pi$ and variance $\sigma^2$. If both agents take risks simultaneously, we assume the risks are independent. That is, agent $i$'s progress up to time $t$ evolves according to $$dk_i=(1+a_i\pi)\,dt+a_i\sigma\, dw_i,$$ where $w_1,w_2$ are independent Wiener processes.\footnote{If one wants progress to be nonnegative, it is equivalent to view $k_i=\log K_i$ as the logarithm of progress, and change the law of motion to $dK_i/K_i=\kh{1+a_i\kh{\pi+\frac{1}{2}\sigma^2}}\,dt+a_i\sigma\, dw_i$.} 

Assume that the current leader succeeds with Poisson rate $\lambda$, at which point the contest ends. That is, conditional on not having succeeded by time $t$, the leader's effort for an additional duration $dt$ yields success during the interval $\kh{t,t+dt}$ with probability $\lambda dt$. In other words, when both agents exert effort, agent $i$'s hazard rate of success is given by $h\kh{k_i,k_{j}}=\lambda\cdot\indic{k_i\ge k_{j}}$.\footnote{Later in Section \ref{sec:design}, we allow the follower to also have a positive hazard rate of success:  $h\kh{k_i,k_{j}}=\overline{\lambda}\cdot\indic{k_i\ge k_{j}}+\underline{\lambda}\cdot\indic{k_i<k_{j}}$ with $\overline{\lambda}>\underline{\lambda} \geq 0$.} 

In the baseline model, the contest is \textit{winner-takes-all}: all prize $P$ is awarded to the winner at the end of the contest.\footnote{Later in Section \ref{sec:design}, we allow the loser to also have a positive prize: the winner gets $\overline{P}$, the loser gets $\underline{P}$, with $\overline{P}\ge \underline{P}\ge 0$.} Thus, agent $i$'s ex post payoff  with realized contest duration $\tau$ is 
\eqn{e^{-r \tau} P \cdot \indic{k_i\kh{\tau}\ge k_{j}\kh{\tau}}-\int_{0}^{\tau} e^{-r t} c d t=e^{-r \tau} P \cdot \indic{k_i\kh{\tau}\ge k_{j}\kh{\tau}}-\frac{c}{r}\left(1-e^{-r \tau}\right).\label{eqn:payoff}}

We focus on (pure strategy) \textit{Markov perfect equilibria (MPE)}.  Following  \cite{MT01}, we define \textit{Markov strategies} as those that depend only on the \textit{payoff-relevant state}; that is, the strategies that are measurable with respect to the \textit{maximally coarse consistent partition} (the \textit{Markov partition}) of histories. Let $\Delta{k}\equiv k_i-k_{j}$ denote the difference in the agents' progress. Applying the \textit{affine invariance criterion}\footnote{It requires to verify that the continuation utility functions are appropriate affine transformations of one another, together with some regularity conditions of non-degeneracy (e.g., in any period, any player can choose different actions to ensure that the opponent's decision problem after that period is affected).} for the Markov partition \citep[Theorem 3.3]{MT01}, one can show that $\Delta{k}$ is indeed the payoff-relevant state of the contest game. Therefore, a Markov strategy  $\kh{\hat{k}_i,s_i}$ for agent $i$ depends only on $\Delta{k}$, and consists of a \textit{stopping rule} $\Delta{k}\le -\hat{k}_i$ specifying \textit{whether} to exert effort, and an \textit{action rule} $s_i\kh{\cdot}:\kh{-\hat{k}_i,\infty}\to \hkh{0,1}$ specifying \textit{how} to exert effort. An MPE $\kh{{k}^*,s^*}=\kh{{k}_i^*,s_i^*,{k}_j^*,s_j^*}$ is a subgame perfect equilibrium in which both players use Markov strategies.

In the next section, we will establish three types of MPEs of the dynamic contest in closed form, and 
show that the structure of the equilibrium strategies crucially depends on the return of the risky move compared to the safe action. Moreover, we will show their uniqueness within a reasonably large family of MPEs, which we will define as \emph{well-behaved} MPEs.


\section{Three Types of Equilibria}\label{sec:MPE}
In this section, we demonstrate three types of MPEs of the dynamic contest and characterize the structure of the equilibrium strategies. In each type of MPE, the follower drops out if the leader takes a significant lead, but engages in ``Hail Marys'' before then: no matter how low the return of the risky move is, the follower undertakes it to try to catch up. 

Moreover, the return of the risky move compared to the safe action is critical in determining the leader's behavior. The analysis boils down to three cases, depending on whether the risky move has a \textit{low}, \textit{medium}, or \textit{high return}. If the risky move has a low return, then the leader will never find it optimal to take it, and the MPE features that each agent takes the risky move if and only if he is the follower. In the opposite case where the risky move has a high return, both agents optimally bear risk as it is a worthwhile investment. Surprisingly, if the risky move has a {medium return}, i.e., in between the previous two cases, the leader will also adopt it when the follower is close to dropping out \textendash{} an interesting preemptive motive.

We denote by $\phi\equiv \lambda P/c$ the \textit{profitability} of the contest, which is the ratio between effective prize and flow cost.  $\phi$ is a crucial parameter that affects equilibrium behavior. Note that $\phi$ needs to be larger than 1 for there to be a meaningful equilibrium, because if $\lambda P<c$, even the leader cannot recover the cost and will drop out immediately. We are now ready to formally define the three cases, i.e., \textit{low/medium/high return}.

\begin{Def}[{Low/medium/high} return]\label{def:case}\
\begin{enumerate}
\item The risky move has a \emph{low return} if $$\pi\le 0.$$
\item The risky move has a \emph{medium return} if $$0<\frac{\pi}{\sigma}<\frac{\sqrt{r+\lambda}}{2}\cdot f\kh{\phi}.$$
\item The risky move has a \emph{high return} if $$\frac{\pi}{\sigma}\ge\frac{\sqrt{r+\lambda}}{2}\cdot f\kh{\phi}.$$	
\end{enumerate}
Here $f\kh{\cdot}$ is a strictly increasing function that rises from 0 to 1 as $\phi$ rises from 1 to infinity.\footnote{The closed-form expression of $f\kh{\cdot}$ is as follows:$$f\kh{\phi}=\frac{\sqrt{2\phi\kh{\sqrt{\phi^2+8}+\phi}-8}}{\sqrt{\phi^2+8}+\phi}.$$}
\end{Def}
In other words, low return means that the risky move has a lower mean return compared with the safe action; medium return means that the risky move has a strictly higher mean return, but the difference is below a certain upper bound; high return means that the difference is above that upper bound.

In the remainder of this section, we show that there exists a different form of MPE in each of the three cases (Propositions \ref{prop:L}, \ref{prop:H} and \ref{prop:M}). Furthermore, these MPEs are unique within a reasonably large family of MPEs, which we refer to as \textit{well-behaved} MPEs (Propositions \ref{prop:Lprime}, \ref{prop:Hprime} and \ref{prop:Mprime}).

\subsection{Equilibrium with Low Return: ``Hail Mary Pass''}
We begin our analysis with the low return case, where the risky move has a lower mean return compared with the safe action ($\pi\le 0$). In the candidate MPE, the follower drops out if the leader takes significant lead, but engages in ``Hail Marys'' before then. That is, no matter how bad the risky move is, the follower undertakes it as an attempt to catch up. In contrast, the leader always selects the safe action, but is prepared to switch to the risky move once caught up by the follower.

Proposition \ref{prop:L} formally states this MPE.
\begin{Prop}[Low return case]\label{prop:L}
If the risky move has a low return, then there exists a unique $k^*>0$ such that the following strategy profile constitutes an MPE:
$$\text{Drop out iff }\Delta{k}\le -k^*, \text{ and otherwise } s_i^*\kh{\Delta{k}}=1\text{ iff }\Delta{k}<0.$$	
\end{Prop}

\begin{proof}
All proofs of the results in the main text are in Appendix \ref{app:proof}.
\end{proof}

The proof Proposition \ref{prop:L} takes two steps. The first step is to convert the (Hamilton-Jacobi-)Bellman equations into second-order ordinary differential equations using It\^{o}'s lemma. The second step imposes the boundary conditions of value matching and smooth pasting, and verifies that the agents have no profitable deviations. 

Figure \ref{fig:L} visualizes this MPE from the perspective of agent $i$. There are four regions: In the leftmost region, agent $i$ drops out because he falls too far behind; conversely, the opponent agent $j$ drops out in the rightmost region. In the middle regions, blue indicates that agent $i$ chooses to take risk, while red indicates that the opponent chooses to take risk. As is shown in Figure \ref{fig:L}, agent $i$ takes the risky move if and only if he falls behind, which mirrors the ``Hail Mary pass'' situation.

\begin{figure}[!htbp]\centering
\includegraphics[width=\textwidth]{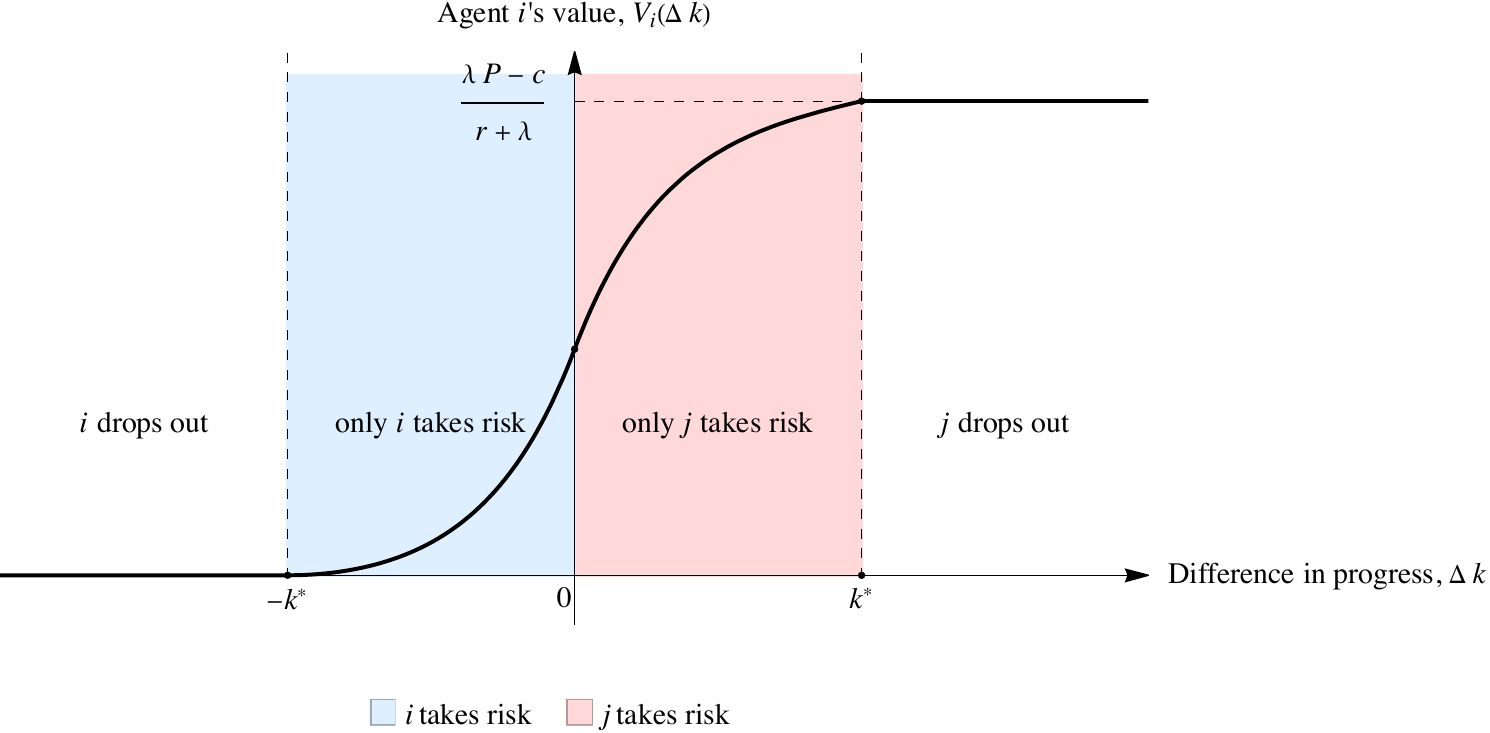}
\caption{Equilibrium with low return: ``Hail Mary pass''}\label{fig:L}
\end{figure}

Note that the MPE in Proposition \ref{prop:L} is completely specified by the stopping boundary $k^*$, which enables the comparative static analysis with respect to the modeling parameters. We summarize the results as follows:

\begin{enumerate}
\item The continuation region becomes larger as the contest becomes more profitable.

Formally, the stopping boundary $k^*$ is strictly increasing in the ratio $P/c$. Moreover, in the limit case, $k^*$ tends to infinity as either the prize of the contest $P$ tends to infinity, or the flow cost $c$ vanishes.

\item The continuation region becomes smaller as the risky move becomes less desirable.

The stopping boundary $k^*$ is strictly increasing in the ratio $\pi/\sigma$ (which by assumption takes a negative value in the low return case). Moreover, in the limit case, $k^*$ tends to zero as either $\pi$ tends to negative infinity, or $\sigma$ vanishes. The only reason the follower chooses to stay is hoping to rely on the ``Hail Mary pass'' to catch up, and the MPE degenerates to \textit{$\epsilon$-preemption}\footnote{ $\epsilon$-preemption is the equilibrium outcome in many standard models of innovation or patent races \citep{FGST83,HV85a,HV85b,LM88}, where a slight advantage causes the opponent to dropout immediately.} if that move gets effectively removed.

\item The result is ambiguous if success occurs at a higher rate.

The stopping boundary $k^*$ changes non-monotonically in the hazard rate of success $\lambda$, because two competing forces coexist. On the one hand, the contest becomes more profitable ($\phi$ increases) as $\lambda$ increases, pushing $k^*$ up. On the other hand, the follower becomes less likely to catch up in time, pushing  $k^*$ down. The first force dominates when $\lambda$ is small, and the second force dominates when $\lambda$ is large,  demonstrating a single-peaked relationship overall.
\end{enumerate}

\subsection{Equilibrium with High Return}
The opposite case is easy to characterize, but perhaps not as interesting. When the risky move has a high return, both agents optimally bear the risk, because it is a worthwhile investment. The formal statement is Proposition \ref{prop:H}.

\begin{Prop}[High return case]\label{prop:H}
If the risky move has a high return, then there exists a unique $k^*>0$ such that the following strategy profile constitutes an MPE:
$$\text{Drop out iff }\Delta{k}\le -k^*, \text{ and otherwise } s_i^*\kh{\Delta{k}}=1.$$	
\end{Prop}
Figure \ref{fig:H} visualizes this MPE from the perspective of agent $i$. Consistent with Figure \ref{fig:L}, blue indicates that agent $i$ chooses to take risk, while red indicates that the opponent chooses to take risk. Both agents choose to take risk in this MPE, and therefore  the entire middle region in Figure \ref{fig:H} is purple. 
\begin{figure}[!htbp]\centering
\includegraphics[width=\textwidth]{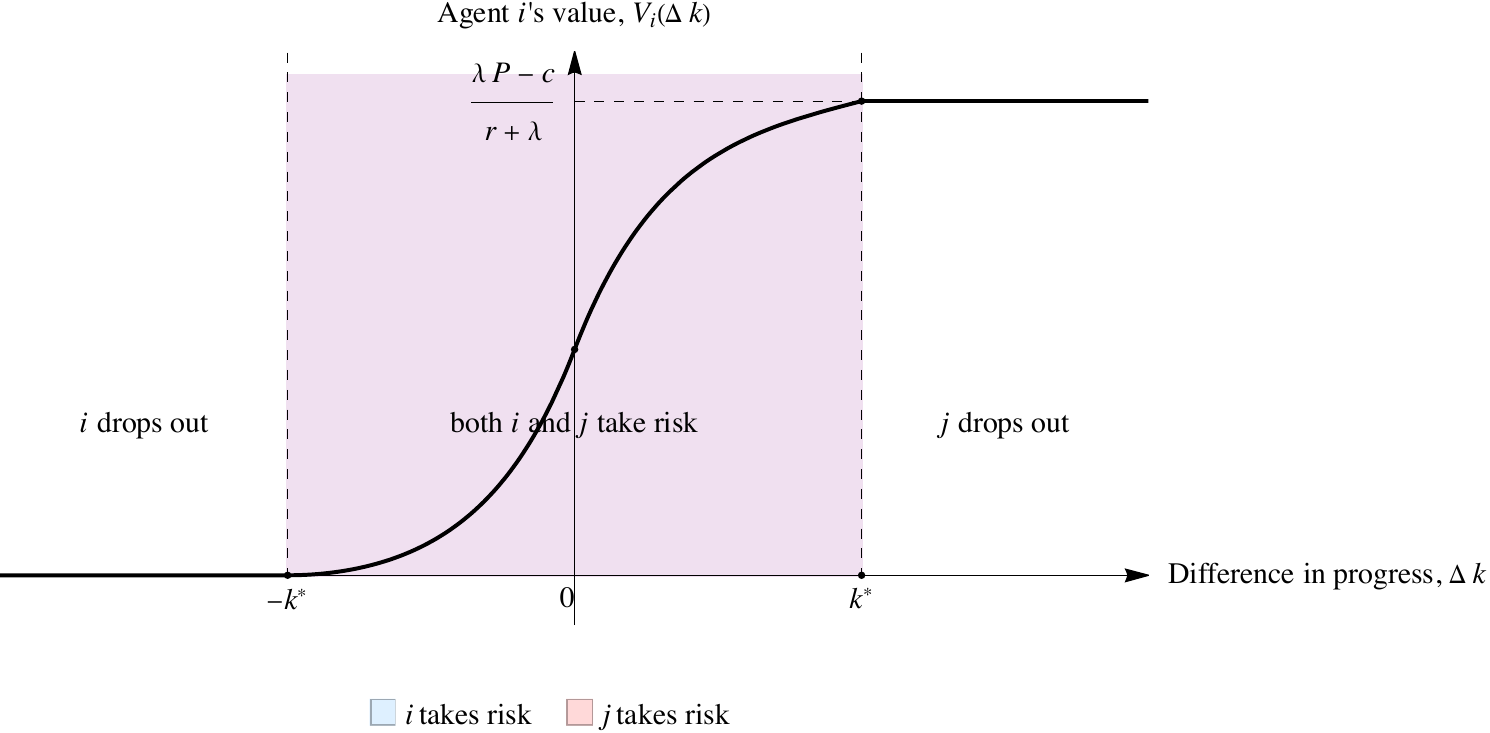}
\caption{Equilibrium with high return}\label{fig:H}
\end{figure}

\subsection{Equilibrium with Medium Return: ``Hail Mary Pass'' $+$ Preemption}
Finally we analyze the remaining medium return case, where the risky move has a higher mean return than the safe action, yet not high enough to be a worthwhile investment in absence of the opponent. 

In the candidate MPE, not only does the follower takes the risky move to try to catch up,  the leader will also adopt it when follower is close to dropping out. The formal statement is Proposition  \ref{prop:M}.

\begin{Prop}[Medium return case]\label{prop:M}
If the risky move has a medium return, then there exists a unique pair of $k^*>k^{**}>0$ such that the following strategy profile constitutes an MPE:
$$\text{Drop out iff }\Delta{k}\le -k^*, \text{ and otherwise } s_i^*\kh{\Delta{k}}=1\text{ iff }(\Delta{k}<0\text{ or }\Delta{k}\ge k^{**}).$$	
\end{Prop}
The MPE strategies outlined in Proposition \ref{prop:M} are combinations of the low and high return cases, and converge appropriately in the limit.\footnote{Formally, $k^{*}-k^{* *} \to  0$ as $\frac{\pi}{\sigma} \downarrow 0$, and the MPE converges to the case of low return. Moreover, $k^{**}\to 0$ as $\frac{\pi}{\sigma} \uparrow \frac{\sqrt{r+\lambda}}{2} \cdot f(\phi)$, and the MPE  converges to the case of high return.}  The analysis reveals two different motivations for taking the risky move: one is the previous \emph{Hail Mary motive},  which the follower use to try to catch up; the other is a new 
\emph{preemptive motive}, where the leader seeks to drive the follower out. The preemptive motive is strongest at the follower's stopping boundary, because the temptation to win the contest for sure is huge for the leader, and in addition the large lead makes it very unlikely that the follower will catch up even if the risky move does not work well.

Figure \ref{fig:M} visualizes this MPE from the perspective of agent $i$. Consistent with Figures  \ref{fig:L} and \ref{fig:H}, blue indicates that agent $i$ chooses to take risk, while red indicates that the opponent chooses to take risk. The purple regions indicate that both agents are taking risks simultaneously, but for different reasons.

\begin{figure}[!htbp]\centering
\includegraphics[width=\textwidth]{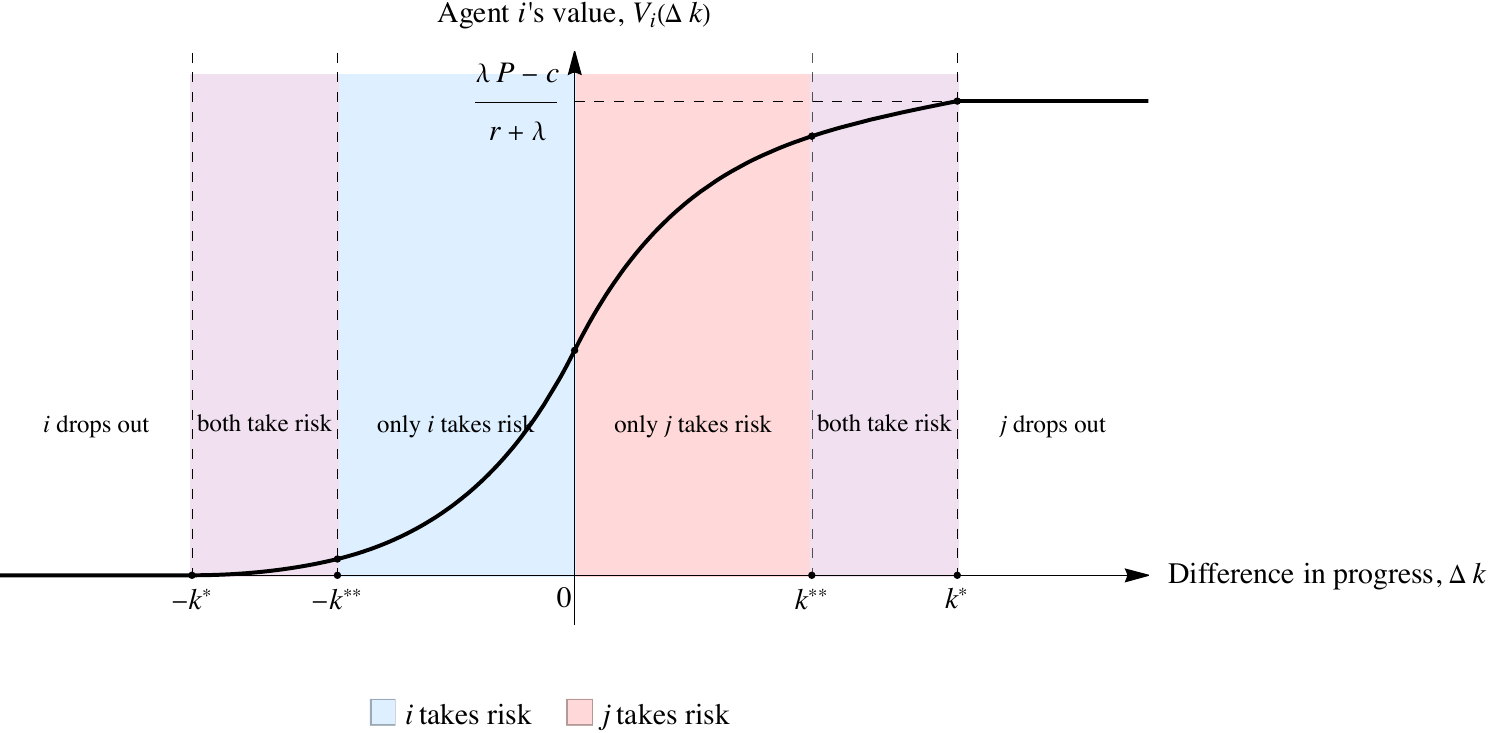}
\caption{Equilibrium with medium return: ``Hail Mary pass'' $+$ Preemption}\label{fig:M}
\end{figure}

\subsection{Uniqueness of the Three Types of Equilibria}

So far, we have identified three types of equilibria, with two motivations for taking the risky move. A   natural question is: Are these equilibria (essentially) unique within their respective case (low/medium/high)? 

In this subsection, we show that the answer is largely yes, that the previously constructed equilibria are unique within a reasonably large family of MPEs. Formally, call an MPE \textit{well-behaved} if it is a solution to the (Hamilton-Jacobi-)Bellman equations of the contest game. We begin by outlining two key observations that help characterize any well-behaved MPE of the contest game.

Suppose that the agents play according to some MPE $\kh{{k}^*,s^*}=\kh{{k}_i^*,s_i^*,{k}_j^*,s_j^*}$. Let $V_i\kh{\Delta{k}}$ denote agent $i$'s value function in the continuation region $-k_{i}^{*}<\Delta k<k_{j}^{*}$, where both agents choose to exert effort and stay in the contest.

The first observation is that the follower always takes the risky move as long as he chooses to stay in the contest.
\begin{Lem}\label{lem:followrisk}
In any well-behaved MPE, the follower always takes the risky move if he chooses to stay in the contest. Formally, $s_i^*\kh{\Delta{k}}=1$ almost everywhere on $\kh{-k^*_i,0}$.
\end{Lem}
To understand the intuition behind Lemma \ref{lem:followrisk}, suppose that the follower instead chooses to stay in the contest and play it safe. One possible response for the leader is to also play it safe, so the difference in the agents' progress will remain constant. Doing so guarantees that the leader wins, so the follower's continuation value will be strictly negative, contradicting the assumption that the follower chooses to stay!

The second observation characterizes the leader's risk attitude: he is strictly risk averse whenever the follower chooses to stay in the contest.
\begin{Lem}\label{lem:leadconcave}
In any well-behaved MPE, the leader's value function is strictly concave. Formally, $V_i''\kh{\Delta{k}}<0$ almost everywhere on $\kh{0,k^*_j}$.
\end{Lem}
The intuition behind Lemma \ref{lem:leadconcave} is as follows. Suppose that the leader's value function is weakly convex over a range. We know from the previous Lemma \ref{lem:followrisk} that the follower takes the risky move, and one possible response of the leader is to also take the risky move. Doing so ensures that the leader is not affected by the mean term, while the variance term helps the leader ($\sigma^2 V_i''\ge 0$). As an implication, the leader's continuation value is at least as good as the value of winning the contest for sure, again contradicting the assumption that follower chooses to stay in!

Combining the above two observations, we greatly simplify the agents' strategies in any well-behaved MPE. Specifically, in any  well-behaved MPE, Lemma \ref{lem:followrisk} shows that the follower always takes the risky move, while Lemma \ref{lem:leadconcave} shows that the leader's behavior switches at most once. Therefore, any well-behaved MPE is completely specified by a pair of stopping boundaries $\left(k_{i}^{*}, k_{j}^{*}\right)$, and possibly with a pair of switching points $\left(k_{i}^{**}, k_{j}^{**}\right)$. By further deriving the best response functions about $\left(k_{i}^{*}, k_{i}^{**}, k_{j}^{*},k_{j}^{**}\right)$ from the boundary conditions of value matching and smooth pasting, we show that
there exists a unique well-behaved MPE in each of the three cases, and it is symmetric. This strengths the previous Propositions \ref{prop:L}, \ref{prop:H} and \ref{prop:M}, respectively, into the following forms.

\begin{Propprime}{prop:L}[Low return case]\label{prop:Lprime}
If the risky move has a low return, then there exists a unique well-behaved MPE and it is symmetric. Formally, the following strategy profile is the unique well-behaved MPE:
$$\text{Drop out iff }\Delta{k}\le -k^*, \text{ and otherwise } s_i^*\kh{\Delta{k}}=1\text{ iff }\Delta{k}<0, \text{ for some }k^*>0.$$
\end{Propprime}

\begin{Propprime}{prop:H}[High return case]\label{prop:Hprime}
If the risky move has a high return, then there exists a unique well-behaved MPE and it is symmetric. Formally, the following strategy profile is the unique well-behaved MPE:
$$\text{Drop out iff }\Delta{k}\le -k^*, \text{ and otherwise } s_i^*\kh{\Delta{k}}=1,\text{ for some }k^*>0.$$	
\end{Propprime}

\begin{Propprime}{prop:M}[Medium return case]\label{prop:Mprime}
If the risky move has a medium return, then there exists a unique well-behaved MPE and it is symmetric. Formally, the following strategy profile is the unique well-behaved MPE: 
$$\text{Drop out iff }\Delta{k}\le -k^*, \text{ and otherwise } s_i^*\kh{\Delta{k}}=1\text{ iff }(\Delta{k}<0\text{ or }\Delta{k}\ge k^{**}),\text{ for some }k^*>k^{**}>0.$$	
\end{Propprime}

In the analysis so far, we have shown that the contest game has a unique well-behaved MPE, and the equilibrium strategies are as characterized in the previous subsections. We are currently working on a natural follow-up question: Are all MPEs well-behaved? In other words, does there exist an MPE that does not satisfy the Bellman equations? While we cannot provide a concrete answer to this question right now, we conclude this section with some conjectures and the possible methods of verifying them.

\cite{Harris93}  develops a theory of stochastic differential games in one dimension in which players' actions are allowed to affect the state in a very general way. The paper proves that all MPEs satisfy the Bellman equations \citep[Theorems 5.6 and 6.5]{Harris93}, as long as the regularity conditions in the paper are satisfied. The baseline model in our paper largely fits the general framework by \cite{Harris93}, with only one exception: if both agents choose to play it safe, then there is completely no variance in the system. This makes the analysis tricky because the order of the associated ODEs  may change suddenly.

In subsequent research, we hope to prove analogous results to \citet[Theorems 5.6 and 6.5]{Harris93}, in order to show that all MPEs in our contest model are well-behaved.


\section{Optimal Prize Allocation}\label{sec:design}
In the previous section, we have focused on a stylized model and identified  two motivations for taking the risky move: the Hail Mary motive and the preemptive motive. This section investigates the impact of such behavior on the optimal prize allocation, by analyzing an extended version of the baseline model that adapts all the derivations and  intuitions in the previous section. The main result for this section is Proposition \ref{prop:optprize}, which shows that ``winner-takes-all'' is optimal only if the designer's budget is small.

\subsection{Extending Baseline Model}\label{subsec:ext}
We first extend the baseline model to consider the role of prize allocation. Assume that the follower in the contest also has a positive hazard rate of success: that is, the leader succeeds at rate $\overline{\lambda}$, and the follower succeeds at rate  $\underline{\lambda}$, with $\overline{\lambda}>\underline{\lambda} \geq 0$. In other words, when both agents exert effort, agent $i$'s hazard rate of success is given by
$$h\kh{k_i,k_{j}}=\Brace{&\overline{\lambda},&&k_i\ge k_j,\\&\underline{\lambda},&&{k_i<k_j}.}$$
Moreover, we allow the loser to also have a positive prize at the end of the contest: the winner gets $\overline{P}$ and the loser gets $\underline{P}$, with $\overline{P}\ge \underline{P}\ge 0$.

The derivations and intuitions in Section \ref{sec:MPE} are directly applicable, as long as we replace $\phi=\lambda P/c$ with the generalized version of the \emph{profitability of contest}:
\eqn{\overline{\phi}\equiv\frac{\kh{\overline \lambda-\underline{\lambda}}\kh{ \overline{P}-\underline{P}}}{c-\kh{\underline\lambda\overline{P}+\overline{\lambda}\underline{P}}}.\label{eqn:prof}}
The definition of the three cases (Definition \ref{def:case}) is adapted as follows:
\begin{Defprime}{def:case}[{Low/medium/high} return]\label{def:caseprime}\
\begin{enumerate}
\item The risky move has a \emph{low return} if $$\pi\le 0.$$
\item The risky move has a \emph{medium return} if\footnote{Recall that the closed-form expression of $f\kh{\cdot}$ is as follows:$$f\kh{\overline\phi}=\frac{\sqrt{2\overline\phi\kh{\sqrt{\overline\phi^2+8}+\overline\phi}-8}}{\sqrt{\overline\phi^2+8}+\overline\phi}.$$} $$0<\frac{\pi}{\sigma}<\frac{\sqrt{r+\overline{\lambda}+\underline{\lambda}}}{2}\cdot f\kh{\overline{\phi}}.$$
\item The risky move has a \emph{high return} if $$\frac{\pi}{\sigma}\ge\frac{\sqrt{r+\overline{\lambda}+\underline{\lambda}}}{2}\cdot f\kh{\overline\phi}.$$	
\end{enumerate}
\end{Defprime}
With this adapted  definition, the three types of MPEs become just as outlined in the previous section, parametrized by $k^*$ and potentially $k^{**}$.

\subsection{A Contest Design Exercise}\label{subsec:design}
Suppose that the designer can choose any pair of prizes $\kh{\overline{P},\underline{P}}$ subject to some fixed budget constraint $\overline{P}+\underline{P}\le B$. Assume that the budget $B$ satisfies $\overline{\lambda} B>c>\underline{\lambda}B$. Note that if $\overline{\lambda} B<c$, the maximum prize that can be awarded to the leader is not enough for him to recover the cost, and thus both agents will drop out immediately; if $\underline{\lambda} B>c$, then both agents will find it worthwhile to stay in the contest forever no matter how the prize is split, so the design problem becomes trivial.

We think it is reasonable for the designer to have any of the following three objectives:
\begin{enumerate}
\item Minimize the expected time needed to achieve success (e.g., innovation race).
\item Maximize the expected time that the follower stays in the contest.
\item Maximize the continuation region (parametrized by $k^*$).
\end{enumerate} 
It turns out that the three objectives are equivalent in our contest model.
\begin{Prop}[Equivalence of objectives]\label{prop:equivobj}
The following objectives of the designer are equivalent:
\begin{enumerate}
\item\label{item:1} Minimize the expected time needed to achieve success.
\item\label{item:2} Maximize the expected time that the follower stays in the contest.
\item\label{item:3} Maximize the continuation region.
\end{enumerate}
\end{Prop}
To understand the proof Proposition \ref{prop:equivobj}, let $\tau_{k^{*}} \equiv \inf \left\{t \geq 0:\abs{k_i\kh{t}-k_j\kh{t}} \geq k^{*}\right\}$ denote the first hitting time to the stopping boundary. In all of the three types of equilibria, $k_i-k_j$ has almost surely continuous sample paths. Therefore, $\tau_{k^*}$ is strictly increasing in $k^*$ in the sense of first-order stochastic dominance, which shows that Objectives \ref{item:2} and \ref{item:3} are equivalent. Moreover, success occurs at rate $\kh{\overline{\lambda}+\underline{\lambda}}$ when both agents stay in the contest, and at a slower rate $\overline{\lambda}$ when the follower drops out. This indicates that Objectives \ref{item:1} and \ref{item:3} are also equivalent.

We are now ready to solve for the designer's optimal prize allocation. Proposition \ref{prop:optprize} shows that {``winner-takes-all''} is optimal only if the designer's budget is {small}.
\begin{Prop}[Optimal prize allocation]\label{prop:optprize}\  
\begin{enumerate}
\item If $\kh{\overline{\lambda}+\underline{\lambda}} B/2<c$, then it is {optimal} to set $\overline{P}=B$ and $\underline{P}=0$.
\item If $\kh{\overline{\lambda}+\underline{\lambda}} B/2> c$, then it is optimal to set $\overline{P}=B/2+\epsilon$ and $\underline{P}=B/2-\epsilon$, with sufficiently small $\epsilon>0$.
\end{enumerate}
\end{Prop}
The intuition behind Proposition \ref{prop:optprize} is as follows. With a small budget,  it is unaffordable to keep both agents in the contest forever. So the designer chooses $\kh{\overline{P},\underline{P}}$ to maximize the continuation region parameterized by $k^*$, which is strictly increasing in the profitability of contest $\overline{\phi}$ defined by equation \eqref{eqn:prof}. Solving this maximization problem yields $\overline{P}=B$ and $\underline{P}=0$, i.e., ``winner-takes-all''.  However, this always induces a finite stopping boundary $k^*$, and becomes suboptimal when there exists some form of allocation that keeps both agents in the contest forever.


\section{Conclusion}\label{sec:concln}
In this paper, we study contests with stochastic progress induced by risky moves.  Agents exert effort and compete to get ahead; they have the option to make progress deterministically, but can also undertake risky moves that stochastically affect their relative position. We demonstrate three types of equilibria, and identify two motivations for taking the risky move: the Hail Mary motive of the follower and the preemptive motive of the leader. Moreover, as an impact of such equilibrium behavior on the optimal prize allocation, ``winner-takes-all'' is optimal only if the designer's budget is small.

\bibliographystyle{myref}
\bibliography{HMP.bib}

\newpage
\begin{appendices}
\setstretch{1}

\renewcommand{\theequation}{\thesection.\arabic{equation}}
\setcounter{equation}{0}
\renewcommand{\theLem}{\thesection.\arabic{Lem}}
\setcounter{Lem}{0}
\renewcommand{\theDef}{\thesection.\arabic{Def}}
\setcounter{Def}{0}
\renewcommand{\theProp}{\thesection.\arabic{Prop}}
\setcounter{Prop}{0}

\section{Proofs of Results in the Main Text}\label{app:proof}
\subsection{Proofs for Section \ref{sec:MPE}}
Suppose that agent $j$ plays according to some Markov strategy $\kh{\hat{k}_j,s_j}$, and fix any Markov strategy of player $i$, $\kh{\hat{k}_i,s_i}$. In the continuation region $-\hat{k}_i<\Delta{k}<\hat{k}_j$, agent $i$'s value function $V_i\kh{\Delta{k}}$ satisfies the following (Hamilton-Jacobi-)Bellman equation (assuming well-behaved) 
$$r V_i\,dt=\Brace{&\max_{a_i}\hkh{\fkh{\lambda\kh{P-V_i}-c}\,dt+\mathbb{E}\fkh{dV_i}},&&0\le \Delta{k}<\hat{k}_j,\\&\max_{a_i}\hkh{\fkh{\lambda\kh{0-V_i}-c}\,dt+\mathbb{E}\fkh{dV_i}},&&-\hat{k}_i<\Delta{k}<0.}$$
Since $\Delta{k}=k_i-k_j$ and $dk_i=(1+a_i\pi)\,dt+a_i\sigma\, dw_i$, It\^{o}'s lemma implies that
\eqns{&dV_i=V_i'\,d k_{i}-V_i'\,d k_{j}+\kh{\frac{1}{2}a_i^2\sigma^2V_{i}''+\frac{1}{2}a_j^2\sigma^2V_{i}''}\,dt\\
\quad\Rightarrow\quad&\mathbb{E}\fkh{dV_i}=\fkh{\kh{a_i-a_j}\pi V_i'+\frac{1}{2}\kh{a_i^2+a_j^2}\sigma^2V_i''}\,dt.}
We can thus rewrite the Bellman equation as 
\eqn{r V_i=\Brace{&\max_{a_i}\hkh{{\lambda\kh{P-V_i}-c}+\kh{a_i-s_j}\pi V_i'+\frac{1}{2}\kh{a_i^2+s_j^2}\sigma^2V_i''},&&0\le \Delta{k}<\hat{k}_j,\\&\max_{a_i}\hkh{{\lambda\kh{0-V_i}-c}+\kh{a_i-s_j}\pi V_i'+\frac{1}{2}\kh{a_i^2+s_j^2}\sigma^2V_i''},&&-\hat{k}_i<\Delta{k}<0,}\label{eqn:Bellman}}
which combines the expected current flow payoff plus the expected change of future payoff due to the drift and volatility of the continuation value. Therefore, agent $i$'s best-responding $s_i^*\kh{\cdot}$ satisfies\eqn{s_i^*=\Brace{&1,&& \pi V_i'+\frac{1}{2}\sigma^2V_i''>0,\\
&\text{indifferent},&&\pi V_i'+\frac{1}{2}\sigma^2V_i''=0,\\
&0,&&\pi V_i'+\frac{1}{2}\sigma^2V_i''<0.}\label{eqn:BRi}}
Agent $i$ solves the differential equation \eqref{eqn:Bellman} by setting \eqn{V_i\kh{\hat{k}_j}=\frac{\lambda P-c}{r+\lambda},\label{eqn:valuematch}}
and \eqn{V_i\kh{-{k}_i^*}=V_i'\kh{-{k}_i^*}=0.\label{eqn:smoothpast}}
These value-matching and smooth-pasting conditions are necessary because it follows from equation \eqref{eqn:payoff} that the value of winning the contest for sure is 
$$\mathbb{E}_{\tau}\left[e^{-r \tau } q-\frac{c}{r}\left(1-e^{-r \tau}\right)\right]=\int_{0}^{\infty}\left[e^{-r \tau } q-\frac{c}{r}\left(1-e^{-r \tau}\right)\right] \lambda e^{-\lambda \tau} d \tau=\frac{\lambda P-c}{r+\lambda},$$ and the value of dropping out immediately is zero.
\begin{proof}[Proof of Proposition \ref{prop:L}]
Consider the candidate equilibrium strategy: $$\text{Drop out iff }\Delta{k}\le -k^*, \text{ and otherwise } s_i^*\kh{\Delta{k}}=1\text{ iff }\Delta{k}<0.$$
Given the candidate equilibrium strategy, we can rewrite the Bellman equation as
$$rV_i=\Brace{&\lambda \kh{P-V_i}-c-\pi V_i'+\frac{1}{2}\sigma^2V_i'',&&0\le\Delta{k}< k^*,\\
&\lambda \kh{0-V_i}-c+\pi V_i'+\frac{1}{2}\sigma^2V_i'',&&-k^*<\Delta{k}<0.}$$
Solving the above second-order linear ODEs by imposing the boundary conditions \eqref{eqn:valuematch} and \eqref{eqn:smoothpast}, we obtain
$$V_i\kh{\Delta{k}}=\Brace{&\frac{\lambda P-c}{r+\lambda}+C_1e^{-\xi_+\Delta{k}}+C_2e^{-\xi_-\Delta{k}},&&0\le\Delta{k}< k^*,\\
&-\frac{c}{r+\lambda}+C_3e^{\xi_+\Delta{k}}+C_4e^{\xi_-\Delta{k}},&&-k^*<\Delta{k}<0,}$$
where $\xi_+>0$ and $\xi_-<0$ are the two roots of the characteristic equation:
\eqn{r=-\lambda+\pi\xi+\frac{1}{2}\sigma^2\xi^2\quad\Rightarrow\quad\xi_\pm=\frac{-{\pi}\pm\sqrt{{\pi}^2+2\kh{r+\lambda}\sigma^2}}{\sigma^2},\label{eqn:xi}}
and $k^*>0$ is the unique positive solution to the first-order condition:
\eqn{\kh{\xi_++\xi_-}^2e^{\kh{\xi_++\xi_-}k^*}-2\xi_+\xi_-\kh{ \xi_+e^{2\xi_+k^*}+\xi_-e^{2\xi_-k^*}}=\phi \kh{\xi_+-\xi_-}\kh{ \xi_+e^{\xi_+k^*}-\xi_-e^{\xi_-k^*}}.\label{eqn:FOC}}

Now we check that each agent does not want to deviate from the candidate equilibrium strategy:
\begin{enumerate}
\item Agent $i$ prefers $a_i=0$ to $a_i=1$ when leading ($\Delta{k}>0$)
$$\pi V_i'+\frac{1}{2}\sigma^2V_i''\le 0.$$
Note that
$$\pi V_i'\kh{k^*}+\frac{1}{2}\sigma^2V_i''\kh{k^*}\le 0\quad\Leftrightarrow\quad \pi\le 0,$$
which holds by assumption. Moreover, the second-order condition $\pi V_i''+\frac{1}{2}\sigma^2V_i'''\ge 0$ holds throughout the range.
\item Agent $i$ prefers $a_i=1$ to $a_i=0$ when following ($\Delta{k}<0$)
$$\pi V_i'+\frac{1}{2}\sigma^2V_i''\ge 0.$$
Note that
$$\pi V_i'\kh{-k^*}+\frac{1}{2}\sigma^2V_i''\kh{-k^*}\ge 0\quad\Leftrightarrow\quad c\ge 0,$$
which holds by assumption. Moreover, the second-order condition $\pi V_i''+\frac{1}{2}\sigma^2V_i'''\ge 0$ holds throughout the range.
\end{enumerate}
This completes the proof.
\end{proof}

\noindent The comparative statics results presented after Proposition \ref{prop:L} can be derived from the first-order condition  \eqref{eqn:FOC}.

\begin{proof}[Proof of Proposition \ref{prop:H}]
Consider the candidate equilibrium strategy: $$\text{Drop out iff }\Delta{k}\le -k^*, \text{ and otherwise } s_i^*\kh{\Delta{k}}=1.$$	
Given the candidate equilibrium strategy, we can rewrite the Bellman equation as
$$rV_i=\Brace{&\lambda \kh{P-V_i}-c+\sigma^2V_i'',&&0\le\Delta{k}< k^*,\\
&\lambda \kh{0-V_i}-c+\sigma^2V_i'',&&-k^*<\Delta{k}<0.}$$
Solving the above second-order linear ODEs by imposing the boundary conditions \eqref{eqn:valuematch} and \eqref{eqn:smoothpast},, we obtain
$$V_i\kh{\Delta{k}}=\Brace{&\frac{\lambda P-c}{r+\lambda}+C_{1} e^{-\eta \Delta k}+C_{2} e^{\eta \Delta k},&&0\le\Delta{k}< k^*,\\
&-\frac{c}{r+\lambda}+C_{3} e^{\eta \Delta k}+C_{4} e^{-\eta \Delta k},&&-k^*<\Delta{k}<0,}$$
where $\eta>0$ is the positive root of the characteristic equation:
\eqn{r=-\lambda+\sigma^2\eta^2\quad\Rightarrow\quad\eta=\frac{\sqrt{r+\lambda}}{\sigma},\label{eqn:eta}}
and $k^*>0$ is the unique positive solution to the first-order condition:\footnote{The quartic equation \eqref{eqn:FOC2} gives rise to the following closed-form solution:
\eqns{k^*=\frac{1}{\eta}\log\frac{\kh{\sqrt{\phi^2+8}+\phi}+\sqrt{2\phi\kh{\sqrt{\phi^2+8}+\phi}-8}}{4}}.}
\eqn{e^{2\eta k^*}+e^{-2\eta k^*}=\phi\left(e^{\eta k^{*}}+e^{-\eta k^{*}}\right). \label{eqn:FOC2}}
Now we check that each agent does not want to deviate from the candidate equilibrium strategy:
\begin{enumerate}
\item Agent $i$ prefers $a_i=1$ to $a_i=0$ when leading ($\Delta{k}>0$)
$$\pi V_i'+\frac{1}{2}\sigma^2V_i''\ge 0.$$
Note that
$$\pi V_i'\kh{k^*}+\frac{1}{2}\sigma^2V_i''\kh{k^*}\ge 0\quad\Leftrightarrow\quad \pi\ge 0,$$
and that
$$\pi V_i'\kh{0+}+\frac{1}{2}\sigma^2V_i''\kh{0+}\ge 0\quad\Leftrightarrow\quad \frac{\pi}{\sigma}\ge\frac{\sqrt{r+\lambda}}{2}\cdot\frac{\sqrt{2\phi\kh{\sqrt{\phi^2+8}+\phi}-8}}{\sqrt{\phi^2+8}+\phi}=\frac{\sqrt{r+\lambda}}{2}\cdot f\kh{\phi},$$
both of which hold by assumption. 
\item Agent $i$ prefers $a_i=1$ to $a_i=0$ when following ($\Delta{k}<0$)
$$\pi V_i'+\frac{1}{2}\sigma^2V_i''\ge 0.$$
Note that
$$\pi V_i'\kh{-k^*}+\frac{1}{2}\sigma^2V_i''\kh{-k^*}\ge 0\quad\Leftrightarrow\quad c\ge 0,$$
which holds by assumption. Moreover, the second-order condition $\pi V_i''+\frac{1}{2}\sigma^2V_i'''\ge 0$ holds throughout the range.
\end{enumerate}
This completes the proof.
\end{proof}

\begin{proof}[Proof of Proposition \ref{prop:M}]
Consider the candidate equilibrium strategy: 
$$\text{Drop out iff }\Delta{k}\le -k^*, \text{ and otherwise } s_i^*\kh{\Delta{k}}=1\text{ iff }(\Delta{k}<0\text{ or }\Delta{k}\ge k^{**}).$$	
Given the candidate equilibrium strategy, we can rewrite the Bellman equation as
$$rV_i=\Brace{&\lambda \kh{P-V_i}-c+\sigma^2V_i'',&&k^{**}\le\Delta{k}< k^*,\\
&\lambda \kh{P-V_i}-c-\pi V_i''+\frac{1}{2}\sigma^2V_i'',&&0\le\Delta{k}< k^{**},\\
&\lambda \kh{0-V_i}-c+\pi V_i''+\frac{1}{2}\sigma^2V_i'',&&-k^{**}\le \Delta{k}<0,\\
&\lambda \kh{0-V_i}-c+\sigma^2V_i'',&&-k^*<\Delta{k}<-k^{**}.}
$$
Solving the above second-order linear ODEs by imposing the boundary conditions \eqref{eqn:valuematch} and \eqref{eqn:smoothpast}, we obtain
$$V_i\kh{\Delta{k}}=\Brace{&\frac{\lambda P-c}{r+\lambda}+C_1e^{-\eta\Delta{k}}+C_2e^{-\eta\Delta{k}},&&k^{**}\le\Delta{k}< k^*,\\
&\frac{\lambda P-c}{r+\lambda}+C_3e^{-\xi_+\Delta{k}}+C_4e^{-\xi_-\Delta{k}},&&0\le\Delta{k}< k^{**},\\
&-\frac{c}{r+\lambda}+C_5e^{\xi_+\Delta{k}}+C_6e^{\xi_-\Delta{k}},&&-k^{**}<\Delta{k}<0,\\
&-\frac{c}{r+\lambda}+C_7e^{\eta\Delta{k}}+C_8e^{-\eta\Delta{k}},&&-k^*<\Delta{k}<-k^{**},}$$
where $\xi_+>0$, $\xi_-<0$ and $\eta>0$ are defined in equations \eqref{eqn:xi} and \eqref{eqn:eta}, $k^*>0$ is the unique positive solution to the first-order condition:
\eqns{V_i'\kh{-k^*}=0,}
and $k^{**}\in\kh{0,k^*}$ is pinned down by the indifference condition
\eqn{\pi V_i'\kh{k^{**}}+\frac{1}{2}\sigma^2V_i''\kh{k^{**}}=0\quad\Rightarrow\quad \frac{e^{\eta \kh{k^*-k^{**}}}-e^{-\eta \kh{k^*-k^{**}}}}{e^{\eta \kh{k^*-k^{**}}}+e^{-\eta \kh{k^*-k^{**}}}}=\frac{\pi}{\sigma}\cdot\frac{2}{\sqrt{r+\lambda}}.\label{eqn:indiff}}
That is,
$$k^{**}=k^*-\frac{1}{2\eta}\log\frac{1+\frac{\pi}{\sigma}\cdot\frac{2}{\sqrt{r+\lambda}}}{1-\frac{\pi}{\sigma}\cdot\frac{2}{\sqrt{r+\lambda}}}.$$
Now we check that each agent does not want to deviate from the candidate equilibrium strategy:
\begin{enumerate}
\item Agent $i$ prefers $a_i=1$ to $a_i=0$ when leading by a lot ($\Delta{k}>k^{**}$)
$$\pi V_i'+\frac{1}{2}\sigma^2V_i''\ge 0,$$
but the other way around when leading slightly ($0<\Delta{k}<k^{**}$)
$$\pi V_i'+\frac{1}{2}\sigma^2V_i''\le 0.$$
This is ensured by the indifference condition \eqref{eqn:indiff} and the second order condition $\pi V_i''+\frac{1}{2}\sigma^2V_i'''\ge 0.$
\item Agent $i$ prefers $a_i=1$ to $a_i=0$ when following ($\Delta{k}<0$)
$$\pi V_i'+\frac{1}{2}\sigma^2V_i''\ge 0.$$
Note that
$$\pi V_i'\kh{-k^*}+\frac{1}{2}\sigma^2V_i''\kh{-k^*}\ge 0\quad\Leftrightarrow\quad c\ge 0,$$
which holds by assumption. Moreover, the second-order condition $\pi V_i''+\frac{1}{2}\sigma^2V_i'''\ge 0$ holds throughout the range.
\end{enumerate}
This completes the proof.
\end{proof}

\begin{proof}[Proof of Lemma \ref{lem:followrisk}]
Suppose, to the contrary, that 	$s_i^*\kh{\Delta{k}}=0$ for some $-{k}^*_i<\Delta{k}<0$. Consider the Bellman equation \eqref{eqn:Bellman} of the leader $j$ at state $k_j-k_i=-\Delta{k}>0$ (where $s_i^*=0$):
$$r V_j=\max_{a_j}\hkh{{\lambda\kh{P-V_j}-c}+a_j\pi V_j'+\frac{1}{2}a_j^2\sigma^2V_j''}.$$
The safe action $a_j=0$ is always feasible, so we obtain the following inequality
$$r V_j\ge {\lambda\kh{P-V_j}-c}\quad\Rightarrow\quad V_j\ge\frac{\lambda P-c}{r+\lambda}.$$
This shows that when $k_j-k_i=-\Delta{k}$, the expected payoff of agent $j$ is at least as good as that of winning the contest for sure. Therefore, the probability of agent $i$ winning the contest at state $k_i-k_j=\Delta{k}<0$ would be zero. Due to the flow cost, the continuation value of agent $i$, $V_i\kh{\Delta{k}}$, would be strictly negative, contradicting the assumption that agent $i$ chooses to stay in and exert effort! 

This completes the proof.
\end{proof}

\begin{proof}[Proof of Lemma \ref{lem:leadconcave}]
Assume that agent $i$ is the leader. At state $k_{i}-k_{j}=\Delta k \in\left(0, k_{i}^{*}\right)$, we know from Lemma \ref{lem:followrisk} that the follower $j$ takes the risky move, i.e., $s^*_j\kh{-\Delta{k}}=1$. Therefore, the Bellman equation of agent $i$ gives
$$rV_i=\max_{a_i}\hkh{{\lambda\kh{P-V_i}-c}+\kh{a_i-1}\pi V_i'+\frac{1}{2}\kh{a_i^2+1}\sigma^2V_i''}.$$
Suppose, to the contrary, that $V_i''\ge 0$. The risky action $a_i=1$ is always feasible, so we obtain the following inequality
$$r V_i\ge {\lambda\kh{P-V_i}-c}+\sigma^2V_i''\ge {\lambda\kh{P-V_i}-c}\quad\Rightarrow\quad V_i\ge\frac{\lambda P-c}{r+\lambda}.$$
Again, this shows that the expected payoff of agent $i$ is at least as good as that of winning the contest for sure. Thus, the probability of the follower $j$ winning the contest at state $k_j-k_i=-\Delta{k}<0$ would be zero. Due to the flow cost, the continuation value of agent $j$, $V_j\kh{-\Delta{k}}$, would be strictly negative, contradicting the assumption that agent $j$ chooses to stay in! 

This completes the proof.
\end{proof}

\begin{proof}[Proof of Proposition \ref{prop:Lprime}]
Suppose that the agents play according to some well-behaved MPE $\kh{{k}_i^*,s_i^*,{k}_j^*,s_j^*}$. 

We first show that the leader always takes the safe move, i.e., $s^*_i\kh{\Delta{k}}=1$ almost everywhere on $\kh{0,k^*_j}$. Fix some $\Delta{k}\in\kh{0,k^*_j}$. Lemma \ref{lem:leadconcave} implies that $V_i''\kh{\Delta{k}}<0$. Now that $\pi\le 0$, so $V_i''\kh{\Delta{k}}<0$ implies that $\pi V_i'\kh{\Delta{k}}+\frac{1}{2}\sigma^2V_i''\kh{\Delta{k}}<0$. Then from the best response condition \eqref{eqn:BRi}, we obtain $s^*_i\kh{\Delta{k}}=1$. 

Moreover, Lemma \ref{lem:followrisk} shows that the follower always takes the risky move if he chooses to stay in. Hence, any well-behaved MPE is completely specified by a pair of stopping boundaries $\kh{k^*_i,k^*_j}$, and each agent takes the risky move if and only if he is the follower ($\Delta{k}<0$).

Given agent $j$'s choice of stopping boundary ${k^*_j}>0$, agent $i$'s value function is given by the following second-order linear ODEs
$$r V_i=\Brace{&{{\lambda\kh{P-V_i}-c}-\pi V_i'+\frac{1}{2}\sigma^2V_i''},&&0\le \Delta{k}<k^*_j,\\&{{\lambda\kh{0-V_i}-c}+\pi V_i'+\frac{1}{2}\sigma^2V_i''},&&-k^*_i<\Delta{k}<0,}$$
subject to the boundary conditions \eqref{eqn:valuematch} and \eqref{eqn:smoothpast}.
Therefore, we obtain
$$V_i\kh{\Delta{k}}=\Brace{
&\frac{\lambda P-c}{r+\lambda}+C_{1i}e^{-\xi_+\Delta{k}}+C_{2i}e^{-\xi_-\Delta{k}},&&0\le \Delta{k}<k^*_j,\\
&-\frac{c}{r+\lambda}+C_{3i}e^{\xi_+\Delta{k}}+C_{4i}e^{\xi_-\Delta{k}},&&-k^*_i<\Delta{k}<0,}$$
where $\xi_+>0$ and $\xi_-<0$ are  defined in equation \eqref{eqn:xi}, and $k_i^*$ is the unique positive solution to the following first-order condition:
\eqn{\kh{\xi_++\xi_-}\kh{\xi_+e^{\xi_-k^*_i+\xi_+k^*_j}+\xi_-e^{\xi_+k^*_i+\xi_-k^*_j}}-2\xi_+\xi_-\kh{ e^{\xi_+\kh{k^*_i+k^*_j}}+e^{\xi_-\kh{k^*_i+k^*_j}}}=\phi \kh{\xi_+-\xi_-}\kh{ \xi_+e^{\xi_+k^*_j}-\xi_-e^{\xi_-k^*_j}},\label{eqn:FOC3}} 
giving rise to a best response function $k^*_i=BR\kh{k^*_j}$. 

Taking limits and derivatives with respect to $BR\kh{\cdot}$, we show that 
\begin{enumerate}
\item The choices of stopping boundaries are \textit{strategic substitutes}: $BR\kh{k^*_j}$ is strictly decreasing in $k^*_j$.
\item  $BR'\kh{k^*_j}\in\kh{-1,0}$ for all $k^*_j>0$. Moreover, $$\lim_{k^*_j\downarrow 0}BR\kh{k^*_j}<\infty\quad\text{and}\quad \lim_{k^*_j\to\infty}BR\kh{k^*_j}>0.$$
\end{enumerate}
As a result, the function $g\kh{x}=BR\kh{BR\kh{x}}-x$ is strictly decreasing, with $g\kh{0}>0$ and $g\kh{\infty}<0$. Therefore, there exists a unique  $k^*\in\kh{0,\infty}$ such that $k^*=BR\kh{BR\kh{k^*}}$. In particular, the uniqueness of $k^*$ implies that $k^*=BR\kh{k^*}.$

Substituting $k_i^*=k_j^*=k^*$ into equation \eqref{eqn:FOC3} simplifies it to equation \eqref{eqn:FOC}. Proposition \ref{prop:L} then implies that each agent does not want to deviate from the candidate equilibrium strategy.

This completes the proof.
\end{proof}

The remaining two propositions, Propositions \ref{prop:Hprime} and \ref{prop:Mprime}, rely on the following two lemmas.

\begin{Lem}\label{lem:leadrisk}
Suppose $\frac{\pi}{\sigma}> \frac{\sqrt{r+\lambda}}{2}$. In any well-behaved MPE, the leader always takes the risky move. Formally, $s_i^*\kh{\Delta{k}}=1$ almost everywhere on $\kh{0,k^*_j}$.
\end{Lem}
\begin{proof}
Suppose, to the contrary, that $s_i^*\kh{k^{**}_i}=0$ for some $k^{**}_i\in\kh{0,k^*_j}$. Lemma \ref{lem:leadconcave} implies that $V_i''\kh{k^{**}_i}<0$. Combining the best response condition \eqref{eqn:BRi} and the assumption that $\frac{\pi}{\sigma}>\frac{\sqrt{r+\lambda}}{2}$, we obtain
\eqn{\pi V_i'+\frac{1}{2}\sigma^2V_i''\le 0\quad\Rightarrow\quad&\frac{V_i'}{-V_i''}\le \frac{\frac{1}{2}\sigma^2}{\pi}<\frac{\pi}{\frac{1}{2}\kh{r+\lambda}}\quad\text{at }\Delta{k}=k^{**}_i\notag\\
\quad\Rightarrow\quad&\pi V_i''+\frac{1}{2}\kh{r+\lambda}V_i'<0\quad\text{at }\Delta{k}=k^{**}_i.\label{eqn:ineq2}}
Moreover, we know from Lemma \ref{lem:followrisk} that the follower $j$ takes the risky move, i.e., $s^*_j\kh{-k^{**}_i}=1$. Therefore, the Bellman equation of agent $i$ gives
$$rV_i={\lambda\kh{P-V_i}-c}-\pi V_i'+\frac{1}{2}\sigma^2V_i''\quad\text{for }\Delta{k}\text{ within a neighborhood of }k^{**}_i.$$
Differentiating both sides of the ODE, we get
$$rV_i'=-\lambda V_i'-\pi V_i''+\frac{1}{2}\sigma^2V_i'''\quad\Rightarrow\quad \kh{r+\lambda}V_i'=\frac{1}{2}\sigma^2V_i'''-\pi V_i''.$$
Substituting into equation \eqref{eqn:ineq2}, we get
$$\pi V_i''+\frac{1}{4}\sigma^2V_i'''-\frac{1}{2}\pi V_i''<0\quad\text{at }\Delta{k}=k^{**}_i,$$
which implies that
$$\pi V_i''+\frac{1}{2}\sigma^2V_i'''<0\quad\text{at }\Delta{k}=k^{**}_i.$$
To sum up, we obtain
$$\pi V_i''+\frac{1}{2}\sigma^2V_i'''<0\quad\text{whenever}\quad\pi V_i'+\frac{1}{2}\sigma^2V_i''\le 0,$$
implying that
$$\pi V_i'+\frac{1}{2}\sigma^2V_i''<0\quad\text{for all }\Delta{k}>k^{**}_i.$$
The best response condition \eqref{eqn:BRi} then implies that $s_i^*\kh{\Delta{k}}=0$ for all $\Delta{k}>k^{**}_i$. Therefore, for $\Delta{k}\in\kh{k^{**}_i,k^*_j}$ agent $i$'s value function is given by the following second-order linear ODE
$$r V_i=\lambda\kh{P-V_i}-c-\pi V_i'+\frac{1}{2}\sigma^2V_i''\text{ subject to }V_i\kh{k^*_j}=\frac{\lambda P-c}{r+\lambda}.$$
In particular, it follows that
$$\lim_{\Delta{k}\uparrow k^*_j}\kh{-\pi V_i'+\frac{1}{2}\sigma^2V_i''}=0.$$
Combining $\pi>0$ and $V_i''< 0$ (using Lemma~\ref{lem:leadconcave} again), we obtain
$$\lim_{\Delta{k}\uparrow k^*_j}V_i'=\lim_{\Delta{k}\uparrow k^*_j}V_i''=0.$$
However, the only solution to the above ODE that has this limiting property is the constant solution $$V_i\kh{\Delta{k}}= \frac{\lambda P-c}{r+\lambda}\quad\text{ for all }\Delta{k}\in\kh{k^{**}_i,k^*_j}.$$
Again, this shows that the expected payoff of agent $i$ is the same as winning the contest for sure. Thus, the probability of the follower $j$ winning the contest at state $k_j-k_i=-k_i^{**}<0$ would be zero. Due to the flow cost, the continuation value of agent $j$, $V_j\kh{-k_i^{**}}$, would be strictly negative, contradicting the assumption that agent $j$ chooses to stay in! 

This completes the proof.
\end{proof}

\begin{Lem}\label{lem:leadswitch}
Suppose $0<\frac{\pi}{\sigma}<\frac{\sqrt{r+\lambda}}{2}$. In any well-behaved MPE, if the leader takes the risky move at some state, then he also takes it at all higher states (i.e., closer to the follower's dropping boundary). Formally, if $s_i^*\kh{k^{**}_i}=1$ for some $k^{**}_i\in\kh{0,k^*_j}$, then $s_i^*\kh{\Delta{k}}=1$ almost everywhere on $\left[k^{**}_i,k^*_j\right)$.
\end{Lem}
\begin{proof}
Assume that agent $i$ is the leader and $s_i^*\kh{k^{**}_i}=1$ for some $k^{**}_i\in\kh{0,k^*_j}$. Lemma \ref{lem:leadconcave} implies that $V_i''\kh{k^{**}_i}<0$. Combining the best response condition \eqref{eqn:BRi} and the assumption that $0<\frac{\pi}{\sigma}<\frac{\sqrt{r+\lambda}}{2}$, we obtain
\eqn{\pi V_i'+\frac{1}{2}\sigma^2V_i''\ge 0\quad\Rightarrow\quad&\frac{V_i'}{-V_i''}\ge \frac{\frac{1}{2}\sigma^2}{\pi}>\frac{\pi}{\frac{1}{2}\kh{r+\lambda}}\quad\text{at }\Delta{k}=k^{**}_i\notag\\
\quad\Rightarrow\quad&\pi V_i''+\frac{1}{2}\kh{r+\lambda}V_i'>0\quad\text{at }\Delta{k}=k^{**}_i.\label{eqn:ineq1}}
Moreover, we know from Lemma \ref{lem:followrisk} that the follower $j$ takes the risky move, i.e., $s^*_j\kh{-k^{**}_i}=1$. Therefore, the Bellman equation of agent $i$ gives
$$rV_i={\lambda\kh{P-V_i}-c}+\sigma^2V_i''\quad\text{for }\Delta{k}\text{ within a neighborhood of }k^{**}_i.$$
Differentiating both sides of the ODE, we get
$$rV_i'=-\lambda V_i'+\sigma^2V_i'''\quad\Rightarrow\quad \kh{r+\lambda}V_i'=\sigma^2V_i'''.$$
Substituting into equation \eqref{eqn:ineq1}, we get
$$\pi V_i''+\frac{1}{2}\sigma^2V_i'''>0\quad\text{at }\Delta{k}=k^{**}_i.$$
To sum up, we obtain
$$\pi V_i''+\frac{1}{2}\sigma^2V_i'''>0\quad\text{whenever}\quad\pi V_i'+\frac{1}{2}\sigma^2V_i''\ge 0,$$
implying that
$$\pi V_i'+\frac{1}{2}\sigma^2V_i''>0\quad\text{for all }\Delta{k}>k^{**}_i.$$
The best response condition \eqref{eqn:BRi} then implies that $s_i^*\kh{\Delta{k}}=1$ for all $\Delta{k}>k^{**}_i$. 

This completes the proof.
\end{proof}

\begin{proof}[Proof of Propositions \ref{prop:Hprime} and \ref{prop:Mprime}]

Assume that $\pi>0$, and that the agents play according to some well-behaved MPE $\kh{{k}_i^*,s_i^*,{k}_j^*,s_j^*}$. 

Denote by $k_i^{**}\equiv\inf\hkh{\Delta{k}\in\kh{0,k^*_j}:s^*_i\kh{\Delta{k}}=1}.$ Lemmas  \ref{lem:leadrisk} and \ref{lem:leadswitch} jointly imply that 
\eqns{s_i^*\kh{\Delta{k}}=\Brace{&1,&&k_i^{**}< \Delta{k}<k^*_j,\\
&0,&&0<\Delta{k}<k_i^{**},}}
regardless of whether $\frac{\pi}{\sigma}<\frac{\sqrt{r+\lambda}}{2}$. Moreover, Lemma \ref{lem:followrisk} shows that the follower always takes the risky move if she chooses to stay in. Hence, any MPE is completely specified by a pair of stopping boundaries $\kh{k^*_i,k^*_j}$ and a pair of switching points $\kh{k^{**}_i,k^{**}_j}$, and each agent $i$ takes the risky move if and only if ($\Delta{k}<0$ or $\Delta{k}\ge k^{**}_i$).

Given agent $j$'s choice of stopping boundary ${k^*_j}$ and switching point $k^{**}_j$, agent $i$'s value function is given by the following second-order linear ODEs
$$r V_i=\Brace{&{{\lambda\kh{P-V_i}-c}+\sigma^2V_i''},&&k_i^{**}\le \Delta{k}<k^*_j,\\
&{{\lambda\kh{P-V_i}-c}-\pi V_i'+\frac{1}{2}\sigma^2V_i''},&&0\le \Delta{k}<k_i^{**},\\
&{{\lambda\kh{0-V_i}-c}+\pi V_i'+\frac{1}{2}\sigma^2V_i''},&&-k^{**}_j<\Delta{k}<0,\\
&{{\lambda\kh{0-V_i}-c}+\sigma^2V_i''},&&-k^*_i<\Delta{k}<-k^{**}_j.}$$
Solving the above second-order linear ODEs by imposing the boundary conditions \eqref{eqn:valuematch} and \eqref{eqn:smoothpast}, we obtain
$$V_i\kh{\Delta{k}}=\Brace{
&\frac{\lambda P-c}{r+\lambda}+C_{1i}e^{-\eta\Delta{k}}+C_{2i}e^{\eta\Delta{k}},&&k^{**}_i\le\Delta{k}< k^*_j,\\
&\frac{\lambda P-c}{r+\lambda}+C_{3i}e^{-\xi_+\Delta{k}}+C_{4i}e^{-\xi_-\Delta{k}},&&0\le\Delta{k}< k^{**}_i,\\
&-\frac{c}{r+\lambda}+C_{5i}e^{\xi_+\Delta{k}}+C_{6i}e^{\xi_-\Delta{k}},&&-k^{**}_j<\Delta{k}<0,\\
&-\frac{c}{r+\lambda}+C_{7i}e^{\eta\Delta{k}}+C_{8i}e^{-\eta\Delta{k}},&&-k^*_i<\Delta{k}<-k^{**}_j,}$$
where $\xi_+>0$, $\xi_-<0$ and $\eta>0$ are defined in equations \eqref{eqn:xi} and \eqref{eqn:eta}. Moreover, If agent $i$'s optimal switching point $k^{**}_i$ is interior, it is pinned down by the indifference condition $\pi V_i'\kh{k^{**}_i}+\frac{1}{2}\sigma^2V_i''\kh{k^{**}_i}=0$, which yields
$$ \frac{e^{\eta \kh{k^*_j-k^{**}_i}}-e^{-\eta \kh{k^*_j-k^{**}_i}}}{e^{\eta \kh{k^*_j-k^{**}_i}}+e^{-\eta \kh{k^*_j-k^{**}_i}}}=\frac{2\pi}{\eta\sigma^2}\quad\Rightarrow\quad k^{**}_i=k^*_j-\frac{1}{2\eta}\log\frac{1+\frac{2\pi}{\eta\sigma^2}}{1-\frac{2\pi}{\eta\sigma^2}}.$$
This implies that agent $i$ best responds with
$$k^{**}_i=\max\hkh{0,\,k^*_j-\frac{1}{2\eta}\log\frac{1+\frac{2\pi}{\eta\sigma^2}}{1-\frac{2\pi}{\eta\sigma^2}}}.$$

Analogous to the proof of Proposition \ref{prop:Lprime}, for fixed $k_j^{*}$ (and thus $k_i^{**}$), $k_i^*$ is the unique positive solution to the first-order condition of smooth-pasting. This gives rise to a best response function $k_i^*=\widehat{BR}\kh{k_j^*}$, which has a unique $k^*\in\kh{0,\infty}$ such that $k^*=\widehat{BR}\kh{\widehat{BR}\kh{k^*}}$. In particular, the uniqueness of $k^*$ implies that $k^*=\widehat{BR}\kh{k^*}.$ Additionally,
$$k^*>\frac{1}{2\eta}\log\frac{1+\frac{2\pi}{\eta\sigma^2}}{1-\frac{2\pi}{\eta\sigma^2}}\quad\Leftrightarrow\quad \frac{\pi}{\sigma}<\frac{\sqrt{r+\lambda}}{2}\cdot \frac{\sqrt{2\phi\kh{\sqrt{\phi^2+8}+\phi}-8}}{\sqrt{\phi^2+8}+\phi}=\frac{\sqrt{r+\lambda}}{2}\cdot f\kh{\phi},$$
so it follows that $k_i^{**}=k_j^{**}=k^{**}$, where $k^{**}>0$ in the medium return case and $k^{**}=0$ in the high return case. Propositions \ref{prop:H} and \ref{prop:M} then imply that each agent does not want to deviate from the candidate equilibrium strategy.

This completes the proof.
\end{proof}

\subsection{Proofs for Section \ref{sec:design}}
\begin{proof}[Proof of Proposition \ref{prop:equivobj}]
In each of the three types of equilibria, let $\tau_{k^{*}} \equiv \inf \left\{t \geq 0:\abs{k_i\kh{t}-k_j\kh{t}} \geq k^{*}\right\}$ denote the first hitting time to the stopping boundary, and let $$
\Psi\kh{\theta ;k^*, k_0}\equiv\mathbb{E}\left[e^{-\theta\tau_{k^*}} \given {k_i\kh{0}-k_j\kh{0}}=k_0\right], \quad \theta > 0
$$
denote its Laplace transform.

Success occurs at rate $\kh{\overline{\lambda}+\underline{\lambda}}$ when both agents stay in the contest, and at a slower rate $\overline{\lambda}$ when the follower drops out. Denote by $T$ the duration before success, which is an exponential random variable with mean $1/\kh{\overline{\lambda}+\underline{\lambda}}$ when both agents choose to stay in, and with mean $1/{\overline{\lambda}}$ when only the leader chooses to stay.

The expected length of time that the follower stays in the contest is given by
\eqns{\mathbb{E}\fkh{\int_{0}^\infty\indic{\tau_{k^*}>t\,\&\,T> t}\,dt\given k_0}&=\mathbb{E}\fkh{\int_{0}^\infty\indic{\tau_{k^*}>t}\cdot\indic{T> t}\,dt\given k_0}\\
&=\int_{0}^\infty\mathbb{P}\kh{\tau_{k^*}>t\given k_0}\cdot e^{-\kh{\overline{\lambda}+\underline{\lambda}}t}\,dt\\
&=\int_{0}^\infty\frac{1- e^{-\kh{\overline{\lambda}+\underline{\lambda}}x}}{\overline{\lambda}+\underline{\lambda}}\,dF_{\tau_{k^*}}\kh{x}\\
&=\frac{1-\Psi\kh{\overline{\lambda}+\underline{\lambda} ;k^*,k_0}}{\overline{\lambda}+\underline{\lambda}}.}
Therefore, Objectives \ref{item:2} and \ref{item:3} are equivalent as long as $\Psi$ is strictly decreasing in $k^*$.

To show that $\Psi$ is strictly decreasing in $k^*$, first note that $k_i-k_j$ has almost surely continuous sample paths in all three types of equilibria. Therefore, $\tau_{k^*}$ is strictly increasing in $k^*$ in the sense of first-order stochastic dominance. Next, observe that $e^{-\theta x}$ is strictly decreasing in $x$. Combining this with the monotonicity of $\tau_{k^*}$, we conclude that $\Psi(\theta;k^*,k_0)=\mathbb{E}\fkh{e^{-\theta\tau_{k^*}}\given k_i(0)-k_j(0)=k_0}$ is strictly decreasing in $k^*$.\footnote{As a sanity check, in the low return case ($\pi\le 0$), $\abs{\Delta{k}}$ follows a reflected Brownian motion with drift $-\pi$ and volatility $\sigma$. It follows from Theorem 1 in \cite{Mayerhofer19} that
\eqns{
\Psi\kh{\theta ;k^*, k_0}=e^{\frac{\kh{-\pi}\kh{k^*-k_0}}{\sigma^{2}}} \frac{\sqrt{\pi^{2}+2 \theta \sigma^{2}} \cosh \left(\frac{k_0 \sqrt{\pi^{2}+2 \theta \sigma^{2}}}{\sigma^{2}}\right)+\kh{-\pi} \sinh \left(\frac{k_0 \sqrt{\pi^{2}+2 \theta \sigma^{2}}}{\sigma^{2}}\right)}{\sqrt{\pi^{2}+2 \theta \sigma^{2}} \cosh \left(\frac{k^* \sqrt{\pi^{2}+2 \theta \sigma^{2}}}{\sigma^{2}}\right)+\kh{-\pi} \sinh \left(\frac{k^*\sqrt{\pi^{2}+2 \theta \sigma^{2}}}{\sigma^{2}}\right)},}
which is indeed strictly decreasing in $k^*$.}

Similarly, the expected time needed to achieve success is given by
\eqns{\mathbb{E}\fkh{\int_{0}^\infty\indic{T> t}\,dt\given k_0}&=\mathbb{E}\fkh{\int_{0}^\infty\indic{\tau_{k^*}>t}\cdot\indic{T> t}\,dt\given k_0}+\mathbb{E}\fkh{\int_{0}^\infty\indic{\tau_{k^*}\le t}\cdot\indic{T> t}\,dt\given k_0}\\
&=\int_{0}^\infty\mathbb{P}\kh{\tau_{k^*}>t\given k_0}\cdot e^{-\kh{\overline{\lambda}+\underline{\lambda}}t}\,dt+\int_{0}^\infty\mathbb{P}\kh{\tau_{k^*}\le t\given k_0}\cdot e^{-\overline{\lambda}t}\,dt\\
&=\int_{0}^\infty\frac{1- e^{-\kh{\overline{\lambda}+\underline{\lambda}}x}}{\overline{\lambda}+\underline{\lambda}}\,dF_{\tau_{k^*}}\kh{x}+\int_{0}^\infty\frac{e^{-\overline{\lambda}x}}{\overline{\lambda}}\,dF_{\tau_{k^*}}\kh{x}\\
&=\frac{1-\Psi\kh{\overline{\lambda}+\underline{\lambda} ;k^*,k_0}}{\overline{\lambda}+\underline{\lambda}}+\frac{\Psi\kh{\overline{\lambda};k^*, k_0}}{\overline{\lambda}}.}
The monotonicity of $\Psi$  in $k^*$ indicates that Objectives \ref{item:1} and \ref{item:3} are also equivalent.

This completes the proof.
\end{proof}

\begin{proof}[Proof of Proposition \ref{prop:optprize}] Proposition \ref{prop:equivobj} implies that the principal's problem is equivalent to the maximization of $\overline{\phi}$ defined by equation \eqref{eqn:prof} subject to the budget constraint $\overline{P}+\underline{P}\le B$.
\begin{enumerate}
\item If $\kh{\overline{\lambda}+\underline{\lambda}} B/2<c$, then for all choices of $\kh{\overline{P},\underline{P}}$ the denominator of $\overline{\phi}$, $c-\kh{\underline\lambda\overline{P}+\overline{\lambda}\underline{P}}$, is strictly positive.  It follows that
\eqns{\frac{\kh{\overline \lambda-\underline{\lambda}}\kh{ \overline{P}-\underline{P}}}{c-\kh{\underline\lambda\overline{P}+\overline{\lambda}\underline{P}}}\le \frac{\kh{\overline \lambda-\underline{\lambda}}B}{c-\underline{\lambda }B}\quad &\Leftrightarrow\quad cB-c\overline{P}+\kh{c-\overline{\lambda}B-\underline{\lambda} B}\underline{P}\ge 0\\
&\Leftarrow\quad cB-c\kh{B-\underline{P}}+\kh{c-\overline{\lambda}B-\underline{\lambda} B}\underline{P}\ge 0\\
&\Leftrightarrow\quad \kh{c-\kh{\overline{\lambda}+\underline{\lambda}} B/2}\underline{P}\ge 0,}
as desired. Equality holds if $\overline{P}=B$ and $\underline{P}=0$.

\item  If $\kh{\overline{\lambda}+\underline{\lambda}} B/2>c$, then the denominator of $\overline{\phi}$, $c-\kh{\underline\lambda\overline{P}+\overline{\lambda}\underline{P}}$, can be negative if $\overline{P}=B/2+\epsilon$ and $\underline{P}=B/2-\epsilon$, with $\epsilon\in\kh{0,\kh{\kh{\overline{\lambda}+\underline{\lambda}} B/2-c}/\kh{\overline \lambda-\underline{\lambda}}}$. Under such choices of $\kh{\overline{P},\underline{P}}$, $\overline{\phi}=\infty$, which is clearly optimal.
\end{enumerate}
This completes the proof.
\end{proof}

\end{appendices}

\end{document}